\documentclass[a4paper, 11pt]{article}
\usepackage[warn]{mathtext}
\usepackage[T2A]{fontenc}
\usepackage[cp1251]{inputenc}
\usepackage[english]{babel}
\usepackage[bindingoffset=0mm,dvips]{geometry}

\usepackage{amsmath}
\usepackage{amssymb}
\usepackage{amsthm}

\usepackage{euscript}

\newcommand{\df}{\emph}

\newcommand{\B}{\mathfrak{B}}
\newcommand{\F}{\EuScript{F}}
\newcommand{\G}{\EuScript{G}}
\newcommand{\K}{\mathcal{C}}
\newcommand{\Bisect}{\mathcal{S}}
\newcommand{\A}{\mathcal{A}}

\newtheorem{theorem}{Theorem}
\newtheorem{lemma}[theorem]{Lemma}
\newtheorem{claim}[theorem]{Claim}
\newtheorem{cor}[theorem]{Corollary}

\newtheorem{Theorem}{Theorem}
\newtheorem{Prop}[Theorem]{Proposition}

\begin{document}

\renewcommand{\refname}{}
\renewcommand{\contentsname}{}

\title{Learning Read-Once Functions\\Using Subcube Identity Queries}
\author{%
Dmitry V. Chistikov \\ \texttt{dd1email@gmail.com}%
\and
Andrey A. Voronenko \\ \texttt{dm6@cs.msu.su}}
\date{Lomonosov Moscow State University\\Faculty of Computational Mathematics and Cybernetics}

\maketitle

\vspace*{-4ex}

\begin{abstract}
We consider the problem of exact identification
for read-once functions over arbitrary Boolean bases.
We introduce a new type of queries (\df{subcube identity} ones),
discuss its connection to previously known ones, and study the complexity
of the problem in question. Besides these new queries, learning algorithms
are allowed to use classic membership ones. We present a technique of modeling
an equivalence query with a polynomial number of membership
and subcube identity ones, thus establishing (under certain conditions)
a polynomial upper bound on the complexity of the problem.
We show that in some circumstances, though, equivalence queries
cannot be modeled with a polynomial number of subcube identity and membership ones.
We construct an example of an infinite Boolean basis with an exponential lower bound
on the number of membership and subcube identity queries required for
exact identification. We prove that for any finite subset of this basis,
the problem remains polynomial.
\end{abstract}

\vspace*{-4ex}

\tableofcontents

\newpage

\section{Introduction}
\label{s:intro}

Imagine a black box with an unknown Boolean function $f$
of variables $X = \{x_1, \ldots, x_n\}$ hidden inside.
Suppose that one has an opportunity to obtain
correct answers to questions of two types:
\begin{itemize}
\item[$(\imath)$]
if all of the variables from $X$ are assigned specific values,
i.\,e., $x_i = a_i$ for all $x_i \in X$,
what value does $f$ have?
\item[$(\imath\imath)$]
if some of the variables from $X$ are assigned specific values,
i.\,e., $x_i = a_i$ for $x_i \in X' \subseteq X$,
is the value of $f$ determined unambiguously?
\end{itemize}
How many questions does one have to ask in order to identify the function
in the box exactly? Clearly, if there is no prior knowledge of $f$,
one cannot do better than ask $2^n$ questions in the worst case. Indeed,
at the beginning the set of all possibilities consists of $2^{2^n}$ functions.
Each question's answer is a single bit, so the height of a (binary)
deterministic decision tree representing one's strategy cannot be less
than $\log_2 2^{2^n} = 2^n$. However, if one knows a priori that $f$
belongs to a certain class $\K$, the problem can become easier.
A counting argument here gives the lower bound of $\log_2\!|\K|$.

In this paper, we consider classes $\K$ of Boolean functions which are
read-once over various bases $\B$ (formal definitions are given
in section~\ref{s:pre}). While questions of type $(\imath)$
(\df{membership queries}) are fairly common for various learning problems
(several settings for read-once functions are discussed
in section~\ref{s:learn}), questions of type $(\imath\imath)$
(\df{subcube identity queries}) appear to have never been considered by researchers
yet. In section~\ref{s:id}, we introduce this new type of queries
formally, define our \df{learning} model in detail
and study the complexity of the considered problem
(one of \df{exact identification}).

Subsection~\ref{ss:def} is devoted to definitions and problem setting.
In subsection~\ref{ss:idmod}, we discuss a connection between
subcube identity queries in our learning model
and two of Valiant's classic \df{necessity} and \df{possibility} queries.
We show that any algorithm using membership and subcube identity queries
can be transformed into an algorithm using necessity and possibility queries,
and vice versa. We also discuss possibility of using polynomial modeling
techniques for another classic type of queries, namely, Angluin's
\df{equivalence} ones. We demonstrate that subcube identity queries
cannot be modeled with a polynomial number of equivalence ones.
In subsection~\ref{ss:fin}, we use membership and subcube identity
queries to simulate an equivalence query with a polynomial overhead.
We show that the considered problem of
exact identification for read-once functions over finite Boolean bases
from a wide class can be solved with a polynomial number of questions
of type $(\imath)$ and $(\imath\imath)$. We also demonstrate that
if a related problem of \df{checking} read-once functions
can be solved polynomially for a finite basis, then
so does the considered problem of exact identification with
membership and subcube identity queries.

In subsection~\ref{ss:inf}, we compare Angluin's learning model,
which uses membership and equivalence queries, with our model,
which uses membership and subcube identity ones. We provide an example
of an infinite Boolean basis and show that an equivalence query
in Angluin's model for this basis can be ``exponentially more powerful''
than membership and subcube identity ones. More formally, we show
that a problem of identifying exactly an unknown function
from a certain set can be solved with a single equivalence query,
but requires an exponentially large number of membership
and subcube identity ones. This means that an equivalence query,
contrary to the results of the previous subsection,
cannot generally be modeled with a polynomial number of
membership and subcube identity ones. This fact also gives an example
of an infinite Boolean basis such that read-once functions
over this basis cannot be identified exactly with a polynomial
number of queries in our model. In subsection~\ref{ss:border},
we prove that for any finite subset of this basis,
this property does not hold and polynomial algorithms
for exact identification of read-once functions exist.

\section{Preliminaries}
\label{s:pre}

\subsection{Basic definitions}
\label{ss:basics}

Suppose $\B$ is a set of Boolean functions.
We shall call $\B$ a \df{basis} and use its functions
to construct formulae. A formula $\F$ over $\B$ is
\mbox{\df{read-once}} if every variable in~$\F$
appears exactly once. A Boolean function
is said to be \df{read-once over~$\B$}
if it can be expressed with a read-once formula over~$\B$.

Read-once functions over $\{\land, \lor, \neg\}$ are commonly
called ``read-once'' without specifying a basis.
Similarly, read-once functions over $\{\land, \lor\}$
are widely known as ``monotone read-once''.
In this paper, though, we shall not use any of these terms.

Suppose $f$ is a Boolean function of variables $X = \{x_1, \ldots, x_n\}$.
A variable $x_i \in X$ is \df{essential} for $f$ iff
there exist two vectors $a$ and $b$ differing only in $i$th component
such that $f(a) \ne f(b)$. All variables that are not essential
are called \df{fictious}.

A \df{partial assignment} $p$ to variables $X$ is
a mapping from $X$ to $\{0, 1, *\}$. We call an assignment
\df{total} iff it takes all variables from $X$ to $\{0, 1\}$
(such an assignment is usually identified with a bit vector of length $|X|$).
A \df{total extension} of a partial assignment $p$ is any
total assignment $a$ such that $p$ and $a$ disagree only
on variables from $p^{-1}(*)$.

Let $f$ be a Boolean function and $p$ a partial assignment
to its variables.
Denote by $f_p$ a \df{projection} function obtained by ``hardwiring''
the values assigned by $p$ to the corresponding inputs
(whenever $p$ takes $x_i$ to $*$, the corresponding input
is left untouched).
In other words, $f_p$ is a function
of variables $X' = p^{-1}(*)$, its domain comprising
exactly $2^{|X'|}$ Boolean vectors of length $|X'|$.
The value of $f_p$ on an input vector $y$ is equal to $f(x)$,
where $x$ is obtained by extending $y$ with values from $p$.
We say that a projection $f_p$ is \df{induced} by
an assignment $p$.

\subsection{The problem of exact identification}
\label{ss:exact}

Consider the problem of \df{learning} described
by Valiant~\cite{val} and Angluin~\cite{a88}.
The goal of learning is \df{exact identification}:
given a black box with an unknown object from a known class $\K$,
one aims to determine which object is hidden in the box.
Knowing a priori that the object belongs to the class $\K$,
one can use \df{queries} to bring out its properties and,
ultimately, to identify it exactly. Queries are answered
by honest and accurate \df{oracles}.

In this paper, we are not interested in time complexity, but focus
our attention on the number of queries performed by algorithms
in the worst case. The algorithms, therefore, can be represented as
deterministic decision trees. Nevertheless, one can easily
check that all the algorithms run in polynomial time
in terms of $n$ (all our objects are Boolean functions
of variables $X = \{x_1, \ldots, x_n\}$ as described below,
so $n$ is the number of variables), when represented as Turing machines.

\section{Learning read-once functions}
\label{s:learn}

We consider a problem of learning
(identifying exactly, see subsection~\ref{ss:exact})
read-once Boolean functions.
This problem has been studied since paper~\cite{val}.
In the setting being considered, a basis $\B$ and
a set of variables $X = \{x_1, \ldots, x_n\}$ are known a priori.
The corresponding class $\K$ of objects being learnt is
a set of Boolean functions (all functions of variables $X$
which are read-once over $\B$), so the object in question
(the \df{target} function $f$)
can be regarded as an unknown concept or property.
This idea gave names to various types of queries,
suggested by Valiant and Angluin.

\subsection{Necessity and possibility queries}
\label{ss:val}

Valiant's approach to learning read-once functions~\cite{val}
suggested using three types of queries.
We shall describe only the first one (the second and the third ones
are related to the notion of Boolean functions' prime implicants).

A \df{necessity query} takes a single partial assignment $p$ as an input.
The result of the query is ``yes'' if $f_p \equiv 1$,
otherwise the result is ``no''.
Valiant also defined a \df{possibility query}, which is dual to
a necessity one. It also takes a single partial assignment $p$ and
returns ``yes'' iff $f_p \not \equiv 0$.

\subsection{Membership and equivalence queries}
\label{ss:ang}

Angluin's approach to learning~\cite{a88} introduced other types of queries.
We shall describe two of them.

\df{Membership queries} allow one to learn
the value of the target function $f$ on a given input.
Such a query takes an input $x$ (a total assignment to variables $X$)
and returns the corresponding value $f(x)$.

\df{Equivalence queries} allow one to determine
whether the target function can be exactly represented
with a given formula. The algorithm presents
a formula $\G$ representing a Boolean function $g$,
and the corresponding oracle determines whether $f$ is equivalent to $g$.
It either outputs ``yes'' or gives a counterexample $y$
such that $f(y) \ne g(y)$. We consider
only \df{proper} equivalence queries, i.\,e. ones restricted to
functions $g$ from the class $\K$ (a more liberal setting could
allow the use of an arbitrary Boolean function here).

One of the major early results in the area of learning read-once functions belongs to
Angluin, Hellerstein and Karpinski. In paper~\cite{ahk}
they describe an algorithm solving the problem
for the basis of conjunction, disjunction and negation,
using $O(n^3)$ membership and $O(n)$ equivalence queries.
Here $n$ is the number of variables, i.\,e., the cardinality of $X$.
In this paper, we shall always measure the number of queries
performed by an algorithm as a function of $n$.

An early generalization~\cite{ahk} of
Angluin, Hellerstein and Karpinski's result
allows the basis $\B$ to contain
arbitrary symmetric threshold functions.
A \df{threshold function} is a one satisfying the condition
\begin{equation*}
f(x_1, \ldots, x_n) = 1
\Leftrightarrow
\alpha_1 x_1 + \ldots + \alpha_n x_n \ge \alpha_0
\end{equation*}
for some real numbers $\alpha_0, \alpha_1, \ldots, \alpha_n$.
If none of $\alpha_1, \ldots, \alpha_n$ is negative,
then $f$ is clearly monotone; if $\alpha_1 = \ldots = \alpha_n$,
then $f$ is symmetric.
The following theorem belongs to
Bshouty, Hancock, Hellerstein and Karpinski~\cite{bhhkthresh}.

\begin{Theorem}
Read-once functions over
the basis of arbitrary symmetric threshold functions
are exactly identifiable
with $O(n^4)$ membership and $O(n)$ equivalence queries.
\end{Theorem}

Further research in this area
culminated in the following theorem
due to Bshouty, Hancock and Hellerstein~\cite{bhhgen}.

\begin{Theorem}
\label{Th:diagbl}
Read-once functions over the basis of
arbitrary constant $l$ fan-in functions
are exactly identifiable
with $O(n^{l + 2})$ membership and $n$ equivalence queries.
\end{Theorem}

\subsection{Subcube parity queries}
\label{ss:subcube}

Paper~\cite{vch0112} suggested studying a related problem
of learning read-once functions with no fictious variables
using subcube queries. The main goal is exact identification
as described in subsection~\ref{ss:exact}, but in this case all essential variables
of the target function are also considered to be known a priori. In the setting
considered in~\cite{vch0112}, a learning algorithm can use
membership queries defined in subsection~\ref{ss:ang} and
\df{subcube parity queries}, which are defined as follows.

Suppose $f$ is an unknown target function and
$X = \{x_1, \ldots, x_n\}$ is the set of all its essential variables.
A \df{subcube parity query} takes a partial assignment $p$ as an input
and yields the parity (sum modulo $2$)
of all values of the induced projection $f_p$ on
its $2^{|X'|}$ possible inputs (here $X' = p^{-1}(*)$).
The term ``subcube parity'' is determined by the observation that
the values $f(x)$ of the target function are summarized over
an $|X'|$-dimensional subcube of the Boolean hypercube $\{0, 1\}^{|X|}$.
This subcube is \df{restricted} by $p$
and consists of all possible inputs for $f_p$.
Note that a membership query is a particular case of
a subcube parity query (for a total assignment $p$).

A basis $\B$ is called \df{projection closed} if any projection
of a function from $\B$ also lies in $\B$. We shall call $\B$
\df{complex} if it is projection closed and contains
conjunction, disjunction and negation functions. For complex bases,
the following criterion determining the power of
subcube parity queries holds true~\cite{vch0112}:

\begin{Theorem}
Suppose $\B$ is a complex basis. Then read-once functions of variables
$X = \{x_1, \ldots, x_n\}$ over $\B$ are exactly identifiable
with a polynomial number of subcube parity queries iff all functions from $\B$
are read-once over the basis $\{\land, \lor, \neg\}$.
If this is the case, there exists an algorithm using $n^2 - n + 1$ queries,
otherwise an exponential number of queries is necessary for exact identification.
\end{Theorem}

\subsection{The problem of checking}
\label{ss:check}

We also need several results of research in a related area,
that of a \df{checking} problem. Suppose $\K$ is a known class
of objects, and one is given a black box
with an unknown object from $\K$. One is also given
a hypothesis that the box contains a certain object $f \in \K$.
One's task is to check whether this hypothesis is true or false.
The class $\K$, object $f$ and available queries all depend
on a specific setting. Note that the order of queries asked by
an algorithm is not important in the checking problem: the task
of the algorithm is simply to check whether all the
answers are correct. Any such algorithm $\A$, therefore, can be
represented by a \df{checking test}
$T_{\A} = \{ \langle q, q(f) \rangle \colon \A \text{\ performs query\ }q\}$,
where $q(f)$ is the result of $q$ when addressed to $f$.
One can see that $T_{\A}$ is simply
a table of input queries and their return values for $f$.

The problem of checking for read-once functions was set up
in paper~\cite{aavmvk11}. The considered class of objects
consists of all read-once functions of variables
$X = \{x_1, \ldots, x_n\}$ over an arbitrary basis $\B$,
and a target function $f$ is known to depend essentially on
all the variables from $X$. The only available queries are
membership ones. A checking test is a set
$T_{\A} = \{ \langle x, f(x) \rangle \colon \A \text{\ asks the value on\ }x\}$.
One may also identify a checking test with a set of inputs
contained in it.

This problem has been studied for various bases, both for
individual functions and in a ``uniform'' setting (determining
the number of queries sufficient for checking any
read-once function of $n$ variables).
For arbitrary finite bases of functions of fan-in at most $l$,
the following approach was suggested.

Take a target read-once function $f$ of variables $x_1, \ldots, x_n$
(as stated above, all the variables are known to be essential).
Let $X'$ be a subset of $X = \{x_1, \ldots, x_n\}$ of size $l$.
Suppose there exists a partial assignment $p$ such that
$p(x_i) = *$ iff $x_i \in X'$ and the projection $f_p$
depends essentially on all variables from $X'$.
In this case the set of all total assignments $a$ extending $p$
is called an \df{essentiality hypercube} for $f$ satisfying
the set of variables $X'$. An \df{$l$-essentiality hypercube set}
for $f$ is any set $H_f$ containing essentiality hypercubes
satisfying every subset $X' \subseteq X$ of size $l$,
whenever this is possible. If for a certain subset $X'$
such a hypercube does not exist, no restriction is imposed on $H_f$.
If an $l$-essentiality hypercube set for $f$
contains essentiality hypercubes for all $\binom{n}{l}$
of $l$-sized subsets of $X$, the target function $f$ is called
\df{$l$-satisfiable}.

Now denote by $B_l$ the basis of all functions of fan-in at most $l$.
The following theorem is proved in~\cite{aavpmi23}:

\begin{Theorem}
Suppose $l$ is an arbitrary natural number, $l \ge 2$.
Let $f$ be an $l$-satisfiable read-once function over $B_l$
and $H_f$ its $l$-essentiality hypercube set.
Then the values of $f$ on vectors from $H_f$
constitute a checking test for $f$ in the basis $B_l$.
\end{Theorem}

Note that under conditions of the theorem,
the cardinality of $H_f$ is at most
$\binom{n}{l} \cdot 2^l = O(n^l)$,
which is polynomial in terms of $n = |X|$.

Unfortunately, for $l = 3$ and greater, there exist read-once
functions over $B_l$ which are not $l$-satisfiable.
The key problem here lies in verifying the following conjecture:

\begin{Prop}[hypercube conjecture]
Suppose $l$ is an arbitrary natural number, $l \ge 2$.
Let $f$ be a read-once function over $\B \subseteq B_l$. Then:
\vspace{-1.5ex}
\begin{description}
\itemsep=0mm
\parskip=0mm
\item[\textup{(strong form)}]
for any $l$-essentiality hypercube set $H_f$ for $f$
the values of $f$ on vectors from $H_f$
constitute a checking test for $f$ in the basis $\B$;
\item[\textup{(weak form)}]
there exists an $l$-essentiality hypercube set $H_f$ for $f$
such that the values of $f$ on vectors from $H_f$
constitute a checking test for $f$ in the basis $\B$.
\end{description}
\end{Prop}

The strong form of this conjecture for all $\B \subseteq B_l$ was proved
for $l = 2$ in paper~\cite{aavmvk11} (the proof is also presented in Appendix,
since main techniques in this area have not been available in English yet),
for $l = 3, 4$ in paper~\cite{aavpmi23} and for $l = 5$ in paper~\cite{vchb5nsk}.
It remains open for $l \ge 6$: neither form is proved for these values of $l$.
Nevertheless, it is known that the strong form of the conjecture holds
true for any finite basis containing no \df{discriminatory} functions.
A function $f$ of variables $X$ is discriminatory if there exists a non-empty
subset $X'$ of $X$ such that all projections $f_a$ for total assignments
$a$ to the variables $X'$ have at least one fictious variable
from $X \setminus X'$ (all variables from $X$ are considered essential
for $f$). All discriminatory functions have at least
$3$ essential variables; all read-once functions over bases without discriminatory
functions are $l$-satisfiable for any $l$. These results and several
other ones can also be found in~\cite{aavpmi23}.

\section{Subcube identity queries}
\label{s:id}

\subsection{Definition and problem setting}
\label{ss:def}

In this paper, we consider the problem of learning
read-once functions in the following setting.
The aim of learning is exact identification, as described in subsection~\ref{ss:exact}.
We do not impose any restrictions on the target function,
similarly to the settings of subsection~\ref{ss:ang} and contrary
to the settings of subsections~\ref{ss:subcube} and \ref{ss:check}.
That is, we do not require all its variables to be essential,
though we still consider the set $X$ of input variables known a priori.
Formally, if one is given a Boolean basis $\B$ and a set of variables $X$,
then the class $\K$ of objects being learnt is the set of all
Boolean functions of variables $X$ which are read-once over $\B$.
Available queries are membership queries, as defined in subsection~\ref{ss:ang},
and \df{subcube identity queries}, which are defined as follows.

An input to a subcube identity query is a partial
assignment $p$ to variables from $X$. The corresponding oracle
determines the induced projection $f_p$, as described in
subsection~\ref{ss:subcube}, and then checks whether $f_p \equiv b$
for either $b = 0$ or $b = 1$. If so, the oracle outputs ``yes'',
otherwise it outputs ``no'', but does not give any further information.

Note that the result of a subcube identity query, unlike that of an equivalence one,
is always a single bit. Also note that if the input projection $p$
is total, then the oracle always outputs ``yes'', so ``zero-dimensional''
queries (i.\,e., those providing total assignments) are of no use.
In fact, we could even change our definition
of the oracle so that it would output $f(p)$ if $p$ is total.
This modified definition would then generalize one of the membership oracle,
similarly to subcube parity case.

Our goal now is to determine the power of subcube identity queries.
In subsection~\ref{ss:idmod}, we demonstrate that
subcube identity queries cannot be modeled with
a polynomial number of equivalence ones. We also discuss a
connection between subcube identity queries and Valiant's
necessity and possibility ones.
In subsection~\ref{ss:fin}, we show that in some circumstances
subcube identity queries can serve as a substitute for
equivalence ones. In subsection~\ref{ss:inf}, though, we provide
an example of a Boolean basis such that this property does not hold.
A known border between polynomial and exponential
complexity of the considered problem is discussed in subsection~\ref{ss:border}.

\subsection{Some remarks on modeling}
\label{ss:idmod}

Note that if membership queries are
not available in a learning model, then subcube identity queries can turn out
significantly more powerful than classic equivalence ones:

\begin{theorem}
\label{th:monid}
The problem of exact identification for non-constant read-once functions
over the basis $\{\land, \lor\}$ can be solved by an algorithm
performing $O(n^2)$ subcube identity queries.
\end{theorem}

\begin{proof}
Note that for all non-constant read-once functions $f$ over $\{\land, \lor\}$,
the value of $f$ on the vector $\mathbf 1 = (1, \ldots, 1)$ is $1$.
One can use an algorithm from~\cite{ahk}, which uses $O(n^2)$ membership queries
to perform exact identification. Since for monotone Boolean functions
$f(a) = 1$ iff $f(a') = 1$ for all vectors $a'$ such that $a \le a' \le \mathbf 1$,
a membership query for $a$ can be simulated with a subcube identity query
for a partial assignment $p$ such that $p(x_i) = 1$ iff $a(x_i) = a_i = 1$
and $p^{-1}(0) = \emptyset$.
\end{proof}

Angluin, Hellerstein and Karpinski proved~\cite{ahk} that
the same problem cannot be solved with any polynomial number of
equivalence queries. Combined with the result of the theorem,
this means that subcube identity queries cannot be modeled with
a polynomial number of equivalence ones. Whether this holds true
in the presence of membership queries, is an open problem.
Possibility of modeling equivalence queries with a polynomial number
of membership and subcube identity ones is considered
in subsections~\ref{ss:fin} and~\ref{ss:inf}.

It must be remarked that subcube identity queries are closely related
to Valiant's necessity and possibility queries. More strictly,
a subcube identity query for a partial assignment $p$
can be modeled with one necessity and one possibility query for $p$.
Indeed, if the necessity query returns ``yes'', then the subcube identity query
should also return ``yes''. If the possibility query returns ``no'',
then the subcube identity query should still return ``yes''.
In all other cases, the subcube identity query should
return ``no''. What's more, both necessity and possibility queries
can be modeled with one subcube identity and one membership query.
If $p$ is total, then modeling is trivial (no subcube identity queries are needed).
In the other case, if a subcube identity query returns ``no'', both those queries should
return ``no''. Otherwise, a membership query for an arbitrary total extension of $p$
allows to decide which of them should return ``yes'' (the other should
return ``no'').

\subsection{Modeling equivalence queries in finite bases}
\label{ss:fin}

In this subsection, we demonstrate that under certain conditions equivalence
queries can be simulated with membership and subcube identity ones.
We describe the technique of modeling in circumstances allowing
only polynomial overhead. The key fact is stated in the following lemma:

\begin{lemma}
\label{model}
Suppose $\B$ is a finite basis for which
hypercube conjecture holds true.
Then an equivalence query for a read-once function over $\B$
can be modeled with $O(n^l)$ membership and
$O(n^l)$ subcube identity queries, where $l$ is maximum fan-in
of functions from $\B$.
\end{lemma}

\begin{proof}
Suppose $f(x_1, \ldots, x_n)$ is a target function and $g$ is
a function supplied to the equivalence oracle.
The oracle needs to check whether $f \equiv g$ and, if so, output ``yes'',
otherwise give a counterexample $y$ such that $f(y) \ne g(y)$.
\par
Note that $g$ is a read-once function over $\B$.
Denote by $g'$ a function obtained from $g$ by eliminating
all its fictious variables.
Since hypercube conjecture holds true for $\B$,
one can construct a checking test $T'$ for $g'$ containing
$O(n^l)$ answer---proof pairs.
Take an arbitrary total assignment $a$ for fictious variables of $g$
and extend all input vectors from $T'$ with $a$. The obtained set of
pairs $\langle x, g(x) \rangle$ constitutes a membership query table $T$
for $g$. We now demonstrate how $T$ can be used to simulate an equivalence
query.
\par
To reach the desired goal, we run a membership query for
each input vector $x$ contained in pairs from $T$. Denote by $b$
a result of the query. Clearly, $b = f(x)$. If for some $x$
we have $b \ne g(x)$, then we output $x$ and terminate the modeling.
Otherwise, since $T'$ is a checking test for $g'$,
we conclude that $g' \equiv f_a$, where $f_a$ is
the corresponding projection.
\par
Now we must check whether the equality $g(x) = f(x)$ holds for all $x$.
For each pair $\langle x', g(x') \rangle$ in $T'$, run a subcube identity
query for a partial assignment $p$ obtained from $x'$ by assigning $*$ to
all variables lacking values. If all such queries give ``yes'' answers,
then $g \equiv f$, so we output ``yes''. Indeed, since $T'$ is a checking test for $g'$,
in this case we know that $g'$ is equivalent to all projections of $f$
induced by partial assignments $a'$ which assign arbitrary
constant values to fictious variables of $g$. This means that all fictious
variables of $g$ are also fictious for $f$, so $f \equiv g$.
Note that in this case $O(n^l)$ membership and $O(n^l)$ subcube identity
queries are used.
\par
Suppose now that a subcube identity query for some $p$ returns ``no''.
In this case we can find a total assignment $a$ such that $f$ and $g$
disagree on $a$, and output a corresponding input vector. The procedure
performing this task is denoted by $\Bisect(p)$ and defined as follows. Let $x_i$ be
a variable such that $p(x_i) = *$. Denote by $p_b$ a partial
assignment obtained from $p$ by changing the value of $p(x_i)$ to $b$.
If such an assignment is total, then we run a membership query for one
of the total extensions of $p$ and determine the input $x$ such that
$f(x) \ne g(x)$. Otherwise, we run a subcube identity query for $p_0$.
If it returns ``no'', we forget about $p$ and run $\Bisect(p_0)$.
If the query returns ``yes'', we run another subcube identity query
for $p_1$. The answer ``no'' makes us run $\Bisect(p_1)$, and
the answer ``yes'' means that projections $f_{p_0}$ and $f_{p_1}$
disagree on all inputs and we can use a single membership query for
choosing one with property $f_{p_b} \not\equiv g_{p_b}$ and going on.
Thus, $\Bisect(p)$ always terminates and requires $O(n)$ queries
for any $n$-variable functions $f$ and $g$.
\par
Note that without loss of generality, $l \ge 1$, otherwise $\B \subseteq \{0, 1\}$
and an equivalence query can be modeled with a single membership query.
Hence, $O(n^l) + O(n) = O(n^l)$, which concludes the proof.
\end{proof}

The main result of this subsection is formulated as follows:

\begin{theorem}
\label{th:model}
Suppose $\B$ is a finite basis for which
hypercube conjecture holds true.
Then read-once functions over $\B$
are exactly identifiable with $O(n^{l + 2})$ membership
and $O(n^{l + 1})$ subcube identity queries,
where $l$ is maximum fan-in of functions from $\B$.
\end{theorem}

\begin{proof}
Applying Lemma~\ref{model} to an algorithm for exact identification
using $O(n^{l + 2})$ membership and $n$ equivalence queries
(see Theorem~\ref{Th:diagbl})
yields a desired algorithm. The number of membership queries
is $O(n^{l + 2}) + n \cdot O(n^l) = O(n^{l + 2})$,
the number of subcube identity queries is $n \cdot O(n^l) = O(n^{l + 1})$.
\end{proof}

For now, we can say that all the conditions are satisfied
in the particular cases described in the following corollary.

\begin{cor}
Suppose $\B$ is a finite basis.
Also suppose that $\B$ contains either no discriminatory functions
or no functions of fan-in $6$ and greater.
Then read-once functions over $\B$
are exactly identifiable with a polynomial number of
membership and subcube identity queries.
\end{cor}

\subsection{Lower bound for one infinite basis}
\label{ss:inf}

In this subsection we consider the basis of arbitrary
monotone threshold functions. Our key argument
refers to learning the functions of the basis themselves.

Note that for each natural $n \ge 2$ and for every real $s$
the following symmetric function is monotone and threshold:
\begin{equation*}
f(x_1, \ldots, x_n) = 1
\Leftrightarrow
x_1 + \ldots + x_n \ge s.
\end{equation*}
Assume $k = \lfloor n / 2 \rfloor$ and $s = k + 1$.
Increasing $k$ coefficients by $\frac{1}{2k}$
and setting $s$ to $k + \frac{1}{2}$
yields a new monotone threshold function, which
disagrees with $f$ on a single input vector
containing exactly $k$ ones.
Let $\K_n$ be the set of all $\binom{n}{k}$ such functions and $f$.

\begin{lemma}
The problem of exact identification of an unknown function from $\K_n$:
\vspace{-1.5ex}
\begin{enumerate}
\renewcommand{\theenumi}{\alph{enumi}}
\renewcommand{\labelenumi}{\textup{(}\theenumi\textup{)}}
\itemsep=0mm
\parskip=0mm
\item can be solved with a single equivalence query, but
\item cannot be solved with less than
 $\binom{n}{k}$ membership and subcube identity queries.
\end{enumerate}
\end{lemma}

\begin{proof}
The first part is straightforward, because
an equivalence query for $f$ solves the problem.
We now use an adversary argument to prove the second part.
If the queries used are all membership, then the desired
is also straightforward. For subcube identity queries, we observe
that the only reasonable ones are those which supply a partial
assignment $p$ allowing a unique total extension $a$ with exactly
$k$ ones. Indeed, if this is not the case, then $p$ itself either
has at least $k + 1$ ones (or at least $n - k + 1$ zeros; in both cases all corresponding
projections $f_p$ are constant) or allows two different total
extensions with $k - 1$ and $k + 1$ ones, respectively (all corresponding
projections are non-constant). Hence, given that the target
function is taken from the set $\K_n$ defined above, each query
can only reveal its value on a single input vector containing $k$ ones.
Thus, if less than $\binom{n}{k}$ queries have been asked,
an imaginary adversary can always conceive of two suitable functions:
the first is $f$ and the second disagrees with $f$ on an input vector
which has not been inquired yet.
\end{proof}

This lemma implies that subcube identity queries do not
possess the same power as equivalence ones.
More precisely, one cannot use modeling techniques
to substitute membership and subcube identity queries
for equivalence ones with a polynomial overhead only.
We also obtain the following statement concerning
exact identification of monotone threshold functions:

\begin{theorem}
\label{th:negative}
Monotone threshold functions of $n$ variables
require at least $\binom{n}{\lfloor n / 2 \rfloor}$
membership and subcube identity queries for exact identification.
\end{theorem}

Since every basis function is read-once by definition, we obtain the following
lower bound on the number of queries needed for solving our main problem:

\begin{cor}
Read-once functions of $n$ variables
over the basis of all monotone threshold functions
require at least $\binom{n}{\lfloor n / 2 \rfloor}$
membership and subcube identity queries for exact identification.
\end{cor}

\subsection{Polynomial vs. exponential complexity border}
\label{ss:border}

In this subsection we discuss a border between
polynomial and exponential complexity for our setting.
Observe that the following statement holds true:

\begin{claim}
No threshold function is discriminatory.
\end{claim}

\begin{proof}
Without loss of generality, take a monotone threshold function
$g(x_1, \ldots, x_n)$. Suppose that
\begin{equation*}
g(x_1, \ldots, x_n) = 1
\Leftrightarrow
G(x_1, \ldots, x_n) \ge 0,
\end{equation*}
where $G(x_1, \ldots x_n) = \alpha_1 x_1 + \ldots + \alpha_n x_n - \alpha_0$
for some non-negative real numbers $\alpha_0, \alpha_1, \ldots, \alpha_n$.
It is sufficient to show that if a variable $x_i$ is fictious for $g$,
then all the variables $x_j$ such that $\alpha_j \le \alpha_i$ are also fictious.
Indeed, once this fact is proved, one may observe that whenever
all the projections $f_a$ of a monotone threshold function $f$
of variables $X$ induced by total assignments to any fixed subset of $X$
have at least one fictious variable, they must also share a common
fictious variable, which then turns out fictious for $f$.
\par
So, suppose that $x_i$ is a fictious variable and $\alpha_j \le \alpha_i$.
Without loss of generality, assume that $i = n - 1$ and $j = n$.
Then for all $x_1, \ldots, x_{n - 2} \in \{0, 1\}$ the following inequalities
hold true:
\begin{equation*}
G(x_1, \ldots, x_{n - 2}, 0, 0) \le
G(x_1, \ldots, x_{n - 2}, 0, 1) \le
G(x_1, \ldots, x_{n - 2}, 1, 0).
\end{equation*}
Since $x_{n - 1}$ is fictious, the leftmost and the rightmost expressions above
are either both negative or both non-negative, and, obviously, so does
the expression in the center. The same reasoning also holds true for inequalities
\begin{equation*}
G(x_1, \ldots, x_{n - 2}, 0, 1) \le
G(x_1, \ldots, x_{n - 2}, 1, 0) \le
G(x_1, \ldots, x_{n - 2}, 1, 1).
\end{equation*}
This means that $g(x_1, \ldots, x_{n - 2}, x_{n - 1}, 0)$ is always equal
to $g(x_1, \ldots, x_{n - 2}, x_{n - 1}, 1)$, regardless of $x_{n - 1} \in \{0, 1\}$.
So, $x_n$ is fictious for $g$, which gives the desired.
\end{proof}

We know now that the infinite basis of arbitrary monotone threshold functions
contains no discriminatory functions, and so hypercube conjecture
holds true for an arbitrary finite subbasis. Hence, since
$\binom{n}{\lfloor n / 2 \rfloor} \sim 2^n / \sqrt{\pi n / 2}$,
we obtain the following border between polynomial and exponential
complexity of exact identification:

\begin{theorem}
\label{th:border}
The problem of exact identification of read-once functions
over the basis of arbitrary monotone threshold functions
requires an exponential number of membership and subcube identity queries
\textup{(}in terms of the number of variables\textup{)},
but the same problem for an arbitrary finite subbasis
can be solved with a polynomial number of queries.
\end{theorem}

\section{Open problems}

We conclude this paper by formulating three open problems
concerning subcube identity queries:
\begin{enumerate}
\item Theorem~\ref{th:monid} reveals that in some cases
subcube identity queries can prove more useful than
equivalence ones. Does the same property hold true when equivalence queries
are ``supported'' by membership ones?
In what circumstances can subcube identity queries be modeled
using equivalence and membership ones with a polynomial overhead only?
\item Theorem~\ref{th:model} establishes an $O(n^{l + 2})$ upper bound on
the number of queries needed for exact identification of read-once functions
over bases of fan-in $l$ and less, for $l \le 5$.
Is this bound tight in terms of $O(\cdot)$ or does there exist a better
algorithm than the one from~\cite{bhhgen} where equivalence queries are
modeled with membership and subcube identity ones?
\item To what degree may the polynomial vs. exponential border of Theorem~\ref{th:border}
be refined? In other words, what is the complexity of exact identification
of read-once functions over infinite bases of monotone threshold functions?
One may be interested, for instance, in a characterization of infinite bases
of monotone threshold functions which allow learning read-once functions
with a polynomial number of membership and subcube identity queries.
\end{enumerate}

\appendix

\section*{Appendix}
\addcontentsline{toc}{section}{Appendix}

We shall prove hypercube conjecture for $l = 2$.
Denote by $B_2$ the basis of all functions of fan-in $2$ or less.
It is trivial to check that any read-once function over $B_2$
can be represented by a read-once formula over the basis
$B_2' = \{\land, \lor, \oplus, \overline\oplus, \neg, 1, 0\}$
(here $\oplus$ is a $XOR$ of $2$ arguments and $\overline\oplus$ is
its negation). We shall represent
formulae as rooted trees with labeled vertices. We shall place
leaves of such a tree at the bottom, and root on the top.
Any read-once formula over $B_2'$ can be transformed
so that its tree would satisfy the following conditions:
\begin{enumerate}
\renewcommand{\theenumi}{\arabic{enumi}}
\renewcommand{\labelenumi}{\theenumi)}
\parskip=0mm
\itemsep=0mm
\item any vertex labeled with ``$1$'' or ``$0$'' must be the only
 vertex in a tree;
\item all leaves are labeled with different variables or their negations
 (\df{literals});
\item all other vertices are labeled with \df{linear} ($\oplus$, $\overline\oplus$)
 or \df{non-linear} ($\land, \lor$) symbols representing corresponding functions
 of fan-in $2$ or greater;
\item adjacent vertices cannot be labeled with identical symbols
 or with different linear symbols;
\item\label{posclause} any vertex $u$ lying directly below (adjacent to) a vertex $v$ labeled
 with a linear symbol cannot be labeled with $\land$ or a negation of a variable.
\end{enumerate}
Any rooted tree satisfying five conditions above is called
a \df{canonical} tree. Any such tree represents a read-once Boolean function
over $B_2$. Conversely, any such function can be represented by
a canonical tree. The uniqueness of such a tree will be proved later.

Let $X = \{x_1, \ldots, x_n\}$ and suppose that
$f$ is a read-once Boolean function of variables $X$ over $B_2$.
An \df{essentiality square} for variables $x_i, x_j \in X$ ($i \neq j$)
is a set of four vectors differing only in $i$th and $j$th components
such that $f$ restricted to the set of these vectors depends essentially
on both $x_i$ and $x_j$. In other words, these four vectors constitute
the set of all total extensions of such a projection $p$ that
$p^{-1}(*) = \{x_i, x_j\}$ and $f_p$ does not have any fictious variables.
An \df{essentiality square set} for $f$ is any set of Boolean vectors of length $n$
containing an essentiality square for all pairs $\{x_i, x_j\} \subseteq X$.
One can easily see that for all such pairs an essentiality square exists.

A \df{glueing} of a canonical tree is a rooted tree
obtained from a canonical tree by performing the following operations:
\begin{enumerate}
\renewcommand{\theenumi}{\arabic{enumi}}
\renewcommand{\labelenumi}{\theenumi.}
\parskip=0mm
\itemsep=0mm
\item Replacing all linear symbols with $0$ and all non-linear symbols with $1$.
\item Contracting adjacent vertices labeled with $1$.
\item Replacing literals of the form $\overline x_i$ with corresponding variables $x_i$.
\end{enumerate}

We also need the following concepts from graph theory.
A graph on vertices $X$ is a \df{cograph} iff
it is reducible to an empty graph on $X$
by repeatedly complementing its connected components.
Suppose that $T$ is a rooted tree with leaves $X$
and no vertices with exactly one child.
Also suppose that non-leaf vertices of $T$
are properly coloured with $0$ and $1$
(no two adjacent vertices have the same colour).
Any tree satisfying these conditions is called a \df{cotree}.
Denote by $\phi(T)$ a graph on vertices $X$
such that $\{x_i, x_j\}$ is an edge in $\phi(T)$
iff the lowest common ancestor of $x_i$ and $x_j$
is coloured with $1$ in $T$.

\begin{claim}[\cite{clb}]
\label{co}
The mapping $\phi$ is a bijection between
the set of all cotrees with leaves $X$ and
the set of all cographs on vertices $X$.
\end{claim}

We shall use the following notation:
\begin{equation*}
x_i^\sigma =
\begin{cases}
x_i, & \text{if\ } \sigma = 1,\\
\overline x_i, & \text{if\ } \sigma = 0.
\end{cases}
\end{equation*}

\begin{lemma}[glueing lemma]
A glueing $\ddot{T}$ of an arbitrary canonical tree
for a read-once function over $B_2$ is
uniquely determined by the values of $f$
on the vectors of any essentiality square set for $f$.
\end{lemma}

\begin{proof}
For an arbitrary canonical tree $T$, its glueing $\ddot{T}$ is unique.
Let $T_1$ be a canonical tree for $f$ and $\ddot{T}_1$ its glueing.
Note that for all $\sigma_i, \sigma_j, \sigma \in \{0, 1\}$
the linearity of a function
$\left(x_i^{\sigma_{i\mathstrut}} \circ x_j^{\sigma_{j\mathstrut}}\right)^\sigma$,
where $\circ \in \{\land, \lor, \oplus, \overline{\oplus}\}$,
coincides with the linearity of a function $x_i \circ x_j$
(in other words, with the linearity of a symbol $\circ$).
This means that an edge $\{x_i, x_j\}$ belongs to the set of edges
of the graph $\phi(\ddot{T}_1)$ iff all essentiality squares for variables $x_i, x_j$
have non-linear projections of $f$. Hence, $\ddot{T} = \ddot{T}_1$,
the glueings of all canonical trees for $f$ are identical, and
the values of $f$ on an essentiality square set uniquely determine $\phi(\ddot{T})$
and, by Claim~\ref{co}, $\ddot{T}$.
\end{proof}

A rooted subtree $T'$ of a canonical tree $T$
is called a \df{fragment} of a canonical tree $T$ iff it
satisfies the following conditions:
\begin{enumerate}
\renewcommand{\theenumi}{\arabic{enumi}}
\renewcommand{\labelenumi}{\theenumi)}
\parskip=0mm
\itemsep=0mm
\item $T'$ has at least one non-leaf vertex;
\item either the root of $T'$ is the root of $T$,
 or the vertex adjacent to the root of $T'$ and lying above it
 is linear;
\item all vertices of $T'$ lie in $T$ below the root of $T'$;
\item all linear vertices from $T$ that are also in $T'$ are leaves in $T'$;
\item all non-linear vertices from $T$ that are also in $T'$ are not leaves in $T'$;
 all their children are in $T'$.
\end{enumerate}

\begin{lemma}[fragment lemma]
Suppose that $f$ is a read-once function over $B_2$ and
one knows a glueing $\ddot{T}$ of a canonical tree $T$.
Also suppose that all children of a vertex $v$ of $\ddot{T}$, which is
labeled with $1$ and corresponds to a fragment $T'$, are leaves in $\ddot{T}$.
Then one can unambiguously reconstruct $T'$ using the values of $f$
on the vectors from an essentiality square set for $f$.
\end{lemma}

\begin{proof}
The reconstruction of $T'$ can be performed in two steps.
At first, we shall reconstruct two variants of leaves' labels.
Consider the leaves labeled with literals
$x_i^{\sigma_{i\mathstrut}}$ and $x_j^{\sigma_{j\mathstrut}}$
($\sigma_i$ and $\sigma_j$ are unknown).
Since Boolean conjunction and disjunction are both monotone,
all projections of $f$ onto any essentiality square for $x_i$ and $x_j$
have the form
$\left(x_i^{\sigma_{i\mathstrut}} \circ x_j^{\sigma_{j\mathstrut}}\right)^\sigma$,
where $\circ \in \{\land, \lor\}$ and $\sigma \in \{0, 1\}$.
Hence, if such a projection is monotone or antimonotone in both its variables,
then $\sigma_i = \sigma_j$, otherwise $\sigma_i \ne \sigma_j$.
This means that the values of $f$ on the vectors from an essentiality
square set determine two possible vectors of $\sigma$'s for leaves of $T'$,
which differ in every single component.
\par
Take any of these vectors and assume that it is the correct one.
Now we can reconstruct the whole unknown fragment.
Consider two leaves of $T'$ labeled with
$x_i^{\sigma_{i\mathstrut}}$ and $x_j^{\sigma_{j\mathstrut}}$,
respectively. Determine the label $\circ \in \{\land, \lor\}$
of the lowest common ancestor of these leaves in $T$. We shall use
the values of $f$ on the corresponding essentiality square.
The associated projection is a conjunction or a disjunction of
$x_i^{\sigma_{i\mathstrut}}$ and $x_j^{\sigma_{j\mathstrut}}$
(or its negation), so there exists such a Boolean vector $\delta = (\delta_1, \delta_2)$
that the values of this projection on all vectors $\gamma \ne \delta$
differ from its value on $\delta$. If the lowest common ancestor
of the considered leaves is labeled with $\land$, it follows that
$\delta = (\sigma_i, \sigma_j)$. Otherwise, if the lowest common ancestor
is labeled with $\lor$, it follows that
$\delta = (\overline\sigma_i, \overline\sigma_j)$.
This means that the unknown fragment can be reconstructed
with the technique of Claim~\ref{co}.
\par
Note that the inverse vector of $\sigma$'s corresponds
to the same fragment tree with dual labels (symbols $\land$ and
$\lor$ are said to be \df{dual} to each other). By De Morgan's laws,
functions represented by these trees are each other's negation.
If the root of $T'$ is also a root of $T$, then
the right tree can be chosen using the value of $f$ on any input.
If this is not the case, the root of $T'$, according to
the clause~\ref{posclause} of the definition of a canonical tree,
cannot be labeled with $\land$, which eliminates one of the variants.
\end{proof}

\begin{theorem}
Let $f$ be a read-once function over $B_2$
and $M_f$ an essentiality square set for $f$.
Suppose that one knows the values of $f$ on all vectors from $M_f$.
Then one can reconstruct a unique canonical tree for $f$.
\end{theorem}

\begin{proof}
At first, one can reconstruct a unique glueing $\ddot{T}$
of a canonical tree $T$ for $f$, using glueing lemma.
Then for each vertex in $\ddot{T}$ labeled with $1$
and having no descendants except for leaves,
one can reconstruct an associated fragment of $T$,
using fragment lemma. Suppose that $\ddot{T}$ contains a vertex $v$
labeled with $0$ such that all its descendants are leaves
(labeled with $x_{i_1}, \ldots x_{i_p}$) and
vertices labeled with $1$ which have already been considered
(with corresponding subtrees representing functions
$f_{j_1}, \ldots, f_{j_q}$). Also suppose that $v$ has not been
considered yet. Perform a substitution
$x_t = x_{i_1} \oplus \ldots \oplus x_{i_p} \oplus
       f_{j_1} \oplus \ldots \oplus f_{j_q}$,
where $t$ is a new natural number, unique for each $v$.
Such a substitution transforms an essentiality square set for $f$
into an essentiality square set for a new function obtained from $f$.
After that, one can continue the reconstruction of
a canonical tree for $f$. If the following steps
prove that the leaf corresponding to $x_t$ should be labeled
with $\overline x_t$, then the associated vertex in $T$
is labeled with $\overline\oplus$, otherwise it is labeled
with $\oplus$. If $v$ is a root vertex in $\ddot{T}$,
then the label of a root vertex in $T$ is determined
by the value of $f$ on any single input.
\end{proof}

\begin{cor}
Every function $f$ which is read-once over $B_2$
has a unique canonical tree.
\end{cor}

\begin{cor}
Suppose $f$ is a read-once function over $B_2$
and $M_f$ is an essentiality square set for $f$.
Then $T_f = \{\langle x, f(x) \rangle \colon x \in M_f\}$
is a checking test for $f$ in the basis $B_2$,
and its cardinality $|T_f|$ is less or equal to
$4 \binom{n}{2} = O(n^2)$.
\end{cor}

\end{document}